\documentclass[a4paper,UKenglish,cleveref, autoref, thm-restate]{lipics-v2021}

\pdfoutput=1 %
\hideLIPIcs  %
\nolinenumbers

\title{A Bicriterion Concentration Inequality and Prophet Inequalities for $k$-Fold Matroid Unions} %

\author{Noga Alon}{Department of Mathematics, Princeton University, Princeton, USA \and Schools of Mathematics and Computer Science, Tel Aviv University, Tel Aviv, Israel}{nalon@math.princeton.edu}{https://orcid.org/0000-0003-1332-4883}{Research supported in part by NSF grant DMS-2154082.} %

\author{Nick Gravin}{Key Laboratory of Interdisciplinary
Research of Computation and Economics, Shanghai University of Finance and Economics, Shanghai, China}{nikolai@mail.shufe.edu.cn}{https://orcid.org/0000-0002-3845-947X}{Research is supported by National Key R \& D Program of China (2023YFA1009500), NSFC grant 61932002, ``the Fundamental Research Funds for the Central Universities in China''.}

\author{Tristan Pollner}{Department of Management Science and Engineering, Stanford University, Stanford, USA}{tpollner@stanford.edu}{}{}

\author{Aviad Rubinstein}{Department of Computer Science, Stanford University, Stanford, USA}{aviad@cs.stanford.edu}{https://orcid.org/0000-0002-6900-8612}{Research supported in part by NSF CCF-1954927, and a David and Lucile Packard Fellowship.}

\author{Hongao Wang}{Department of Computer Science, Purdue University, West Lafayette, USA}{wang5270@purdue.edu}{https://orcid.org/0000-0003-2025-2443}{}

\author{S. Matthew Weinberg}{Department of Computer Science, Princeton University, Princeton, USA}{smweinberg@princeton.edu}{https://orcid.org/0000-0001-7744-795X}{Research supported in part by NSF CCF-1955205. During Professor Weinberg’s development of this paper, he participated as an expert witness on behalf of the State of Texas in ongoing litigation against Google (the ``Google Litigation'').}

\author{Qianfan Zhang}{Department of Computer Science, Princeton University, Princeton, USA}{qianfan@princeton.edu}{https://orcid.org/0000-0003-3737-1545}{Research supported in part by NSF CCF-1955205.}

\authorrunning{N. Alon, N. Gravin, T. Pollner, A. Rubinstein, H. Wang, S. M. Weinberg, and Q. Zhang} %

\Copyright{Noga Alon, Nick Gravin, Tristan Pollner, Aviad Rubinstein, Hongao Wang, S. Matthew Weinberg, and Qianfan Zhang} %

\ccsdesc[500]{Theory of computation~Online algorithms}

\keywords{Prophet Inequalities, Online Contention Resolution Schemes, Concentration Inequalities} %

\category{} %

\relatedversion{} %

\acknowledgements{The authors are grateful to the anonymous reviewers for helpful feedback on the initial submission of this work.}%

\EventEditors{Raghu Meka}
\EventNoEds{1}
\EventLongTitle{16th Innovations in Theoretical Computer Science Conference (ITCS 2025)}
\EventShortTitle{ITCS 2025}
\EventAcronym{ITCS}
\EventYear{2025}
\EventDate{January 7--10, 2025}
\EventLocation{Columbia University, New York, NY, USA}
\EventLogo{}
\SeriesVolume{325}
\ArticleNo{18}

\usepackage{amsmath,amssymb,amsthm}
\usepackage{mathtools}
\usepackage{graphicx}
\usepackage{aligned-overset}
\usepackage{rotating}
\usepackage{makecell}
\usepackage{booktabs}
\usepackage{multirow}
\usepackage{comment}
\usepackage{derivative}
\usepackage{soul}

\usepackage{algorithm, algpseudocodex}

\declaretheorem[sibling=theorem]{fact}

\renewcommand{\Pr}{\operatorname*{\mathbf{Pr}}}
\DeclareMathOperator*{\E}{\mathbf{E}}

\DeclareMathOperator{\Ber}{\mathrm{Ber}}

\DeclareMathOperator{\OPT}{OPT}

\DeclareMathOperator{\Span}{span}
\DeclareMathOperator{\conv}{ConvexHull}
\DeclareMathOperator{\rank}{rank}
\DeclareMathOperator{\girth}{girth}

\newcommand{\vb}{\boldsymbol}

\newcommand{\cD}{\mathcal{D}}
\newcommand{\cF}{\mathcal{F}}
\newcommand{\cI}{\mathcal{I}}

\newcommand{\cM}{\mathcal{M}}
\newcommand{\cP}{\mathcal{P}}
\newcommand{\cX}{\mathcal{X}}

\newcommand{\occ}{\omega}

\begin{document}

\maketitle

\begin{abstract}
We investigate prophet inequalities with competitive ratios approaching $1$, seeking to generalize $k$-uniform matroids. We first show that large girth does \emph{not} suffice: for all $k$, there exists a matroid of girth $\geq k$ and a prophet inequality instance on that matroid whose optimal competitive ratio is $\frac{1}{2}$. Next, we show $k$-fold matroid unions \emph{do} suffice: we provide a prophet inequality with competitive ratio $1-O(\sqrt{\frac{\log k}{k}})$ for any $k$-fold matroid union. Our prophet inequality follows from an online contention resolution scheme.

The key technical ingredient in our online contention resolution scheme is a novel bicriterion concentration inequality for arbitrary monotone $1$-Lipschitz functions over independent items which may be of independent interest. Applied to our particular setting, our bicriterion concentration inequality yields ``Chernoff-strength'' concentration for a $1$-Lipschitz function that is not (approximately) self-bounding.

\end{abstract}

\section{Introduction}
\label{sec:intro}
Prophet inequalities are fundamental problems in optimal stopping theory, whose study dates back to seminal work of Krengel and Sucheston~\cite{KrengelS78}, and that have wide applications across Economics and Computer Science (e.g., \cite{ChawlaHMS10, DuttingFKL20}). A prophet inequality instance contains a ground set $E$ of elements, a family $\mathcal{F}\subseteq 2^E$ of feasible sets, and a collection of distributions $\{\mathcal{D}_e\}_{e \in E}$. For one element at a time, a random variable $v_e$ is drawn from distribution $\mathcal{D}_e$ independently and revealed to a gambler, who immediately and irrevocably decides whether to accept or reject $e$. The gambler must at all times maintain the set of accepted elements $A \in \cF$, and gets payoff $\sum_{e \in A} v_e$ at the end of the game. A prophet inequality is $c$-competitive if it guarantees $\E[\sum_{e \in A} v_e] \geq c \cdot \E[\max_{S \in \mathcal{F}} \sum_{e \in S} v_e]$.\footnote{The expectation is taken with respect to the random variables $\{v_e\}_{e \in E}$, which in turn makes $A$ a random variable.}

Krengel and Sucheston's seminal result establishes a $\frac{1}{2}$-competitive prophet inequality for any instance where $\mathcal{F}$ is a $1$-uniform matroid (i.e.~at most one element is feasible to accept), and moreover establish that no better guarantee is possible.\footnote{That is, there exist prophet inequality instances over $1$-uniform matroids for which better than a $\frac{1}{2}$-competitive ratio is impossible. The hard instance is quite simple: $v_1 \sim \cD_1$ is a point mass at $1$, and $v_2 \sim \cD_2$ is equal to $\frac{1}{\varepsilon}$ with probability $\varepsilon$ and $0$ otherwise. A gambler who sees $v_1$ first cannot achieve expect reward exceeding $1$, but a prophet who always takes the maximum can achieve expected reward of $2-\varepsilon$.}
For $k$-uniform matroids, however, a significantly improved guarantee of $1-O(\frac{1}{\sqrt{k}})$ is possible~\cite{Alaei14, JWZ22, DW24}. This motivates the following question: for a given $\varepsilon > 0$, what conditions on $\mathcal{F}$ suffice for a $(1-\varepsilon)$-competitive prophet inequality?

\paragraph*{Main Result I: Large Girth does not Suffice.} A natural starting point to address this question is to first understand what makes $k$-uniform matroids ``special'' in the sense that the canonical hard instance cannot be embedded. One conjecture might be because $k$-uniform matroids have large \emph{girth}: there are \emph{no infeasible sets of size $\le k$}. So, a natural first question to ask is whether $\mathcal{F}$ having large girth suffices in order to conclude that any instance over $\mathcal{F}$ admits a $c$-competitive prophet inequality. Our first main result establishes that large girth does \emph{not} suffice.

\begin{restatable}{theorem}{restategirth}
\label{thm:girth}
For all $k \geq 1$ and $\varepsilon > 0$, there exists a prophet inequality instance $(E, \mathcal{F},\{\mathcal{D}\}_{e \in E})$ such that: (a) $(E,\mathcal{F})$ is a graphic matroid with girth $k$, and (b) $(E, \mathcal{F}, \{\mathcal{D}\}_{e \in E})$ does not admit a $(\frac{1}{2}+\varepsilon)$-competitive prophet inequality.
\end{restatable}

Our construction leverages dense graphs of high girth (and a particular construction of~\cite{LUW95}) in order to effectively embed multiple copies of the canonical hard $1$-uniform instance. See \autoref{sec:largegirth} for further details.

\paragraph*{Main Result II: $k$-fold Matroid Unions Suffice.} \autoref{thm:girth} motivates richer generalizations of $k$-uniform matroids. We next consider \emph{$k$-fold matroid unions}, observing that $k$-uniform matroids are the union of $k$ $1$-uniform matroids. Given a matroid $\mathcal{M}=(E,\cF)$ over ground set $E$ with feasible sets $\mathcal{F}$, the \emph{$k$-fold union of $\cM$} is a new matroid $\mathcal{M}^k$ with ground set $E$ and feasible sets $\mathcal{F}^k:=\{F_1 \cup F_2 \cup\cdots \cup F_k\ :\ F_1, F_2,\ldots, F_k \in \mathcal{F}\}$. That is, a set is feasible in $\mathcal{M}^k$ if it can be partitioned into $k$ sets that are each feasible in $\mathcal{M}$.

\begin{restatable}{theorem}{restatekfoldunion}
\label{thm:kfoldunion}
For every prophet inequality instance $(E, \mathcal{F}^k, \{\mathcal{D}_e\}_{e \in E})$ where $(E,\mathcal{F}^k)$ is the $k$-fold union of a matroid $(E,\mathcal{F})$, there exists a $(1-O(\sqrt{\frac{\log k}{k}}))$-competitive prophet inequality.
\end{restatable}

Our proof of \autoref{thm:kfoldunion} follows from a novel Online Contention Resolution Scheme (OCRS). An OCRS is parameterized by a ground set $E$, a feasibility family $\mathcal{F}$, and a vector of probabilities $\vb{x} \in \conv(\{\mathbf{1}_{F} : F \in \mathcal{F}\}) \subseteq [0,1]^E$ (that is, $\vb{x}$ can be written as a convex combination of indicator vectors of feasible sets). One at a time, elements of $E$ are revealed and \emph{active} with probability $x_e$ independently. If an element is active, it can be accepted or rejected (if inactive, it must be rejected), and the accepted elements must at all times be in $\mathcal{F}$. An OCRS is $c$-selectable if every element $e$ is accepted with probability at least $c\cdot x_e$. In this language, \autoref{thm:kfoldunion} follows from a novel $(1-O(\sqrt{\frac{\log k}{k}}))$-selectable OCRS for $k$-fold matroid unions.

To prove our OCRS, we follow a similar framework as~\cite{FSZ16}, and design a recursive decomposition of $\mathcal{F}$ over which to greedily accept active elements. There are two key challenges to applying their framework, which we overview in greater detail in \autoref{subsec:overview-our-ocrs}. We give a representative example below. 

Applied to the $1$-uniform matroid, the~\cite{FSZ16} algorithm simply proposes ``accept any active element independently with probability $b$.'' Then, linearity of expectation suffices to observe that there are at most $b$ elements in expectation that are both active and accepted,\footnote{There are at most $1$ elements in expectation that are active, and each active element is accepted with probability $b$.} and Markov's inequality suffices to guarantee that with probability at least $1-b$, no elements are accepted at all. This suffices to guarantee that for all $e$: (a) with probability at least $1-b$ it is feasible to accept $e$ when revealed, and (b) independently, we will accept $e$ with probability $b$ conditioned on $e$ being active and feasible. This implies a $b(1-b)$-selectable algorithm, which is optimized at $b=\frac{1}{2}$. 

Applied to the $k$-uniform matroid, a natural algorithm would again be ``accept any active element independently with probability $b$.'' Then, linearity of expectation still suffices to observe that there are at most $bk$ elements in expectation that are both active and accepted, but Markov's inequality only guarantees that with probability at least $1-b$, at most $k$ elements are accepted. This would lead to the same $\frac{1}{4}$-selectable OCRS, which is not the desired $1-O(\sqrt{\frac{\log k}{k}})$. Of course, the obvious fix is to use a significantly stronger concentration inequality than Markov's. E.g., a Chernoff bound suffices to guarantee that with probability at least $1-\frac{1}{k}$ at most $k-1$ elements are accepted, when $b = 1-O(\sqrt{\frac{\log k}{k}})$. This leads to the desired $1-O(\sqrt{\frac{\log k}{k}})$ selectable OCRS for \emph{$k$-uniform} matroids. However, Chernoff bounds are insufficient for the general class of $k$-fold matroid unions -- the probability that a particular element is feasible to accept is a highly combinatorial function that depends on the underlying matroid structure. Thus our
\autoref{thm:kfoldunion} has two components: first, a decomposition that reduces the OCRS problem to a concentration inequality and second, a novel concentration inequality, which is our third main result.

\paragraph*{Main Result III: A Bicriterion Concentration Inequality.} Putting aside prophet inequalities for a moment, concentration inequalities are a core aspect of applied probability with widespread application across many areas of Computer Science. One representative setting is the following: Let $f:\{0,1\}^E \rightarrow \mathbb{R}$ be some function, and let $\vb{X} = \langle X_e\rangle_{e \in E}$ be a vector of independent Bernoulli random variables, where $X_e \sim \text{Ber}(p_e)$. A canonical question asks: what is the probability that $f(\vb{X})$ exceeds $\E[f(\vb{X})]+t$?

On one extreme, McDiarmid's inequality holds whenever $f$ is $1$-Lipschitz. On the other, Chernoff bounds are \emph{significantly} stronger, if $f$ is linear (and $1$-Lipschitz). In between, ``Chernoff-strength'' concentration holds whenever $f$ is fractionally-subadditive or (approximately) self-bounding~\cite{BLM00,BLM03,MR06,BLM09,Von10}, but this provably does not extend even to the case when $f$ is subadditive~\cite{Von10}. 

Our third main result provides a \emph{bicriterion concentration inequality} for any monotone $1$-Lipschitz function. Specifically, if $\vb{X}$ is a vector of Bernoulli random variables with probability vector $\vb{p}$, let $\vb{X}^{(s)}$ denote a vector of Bernoulli random variables with probability vector $e^{-s}\vb{p}$. That is, each probability $p_i$ has been decreased by a factor of $e^{-s}$. Our new concentration inequality establishes:

\begin{restatable}{theorem}{restateconcentration}
\label{thm:newconcentration}
Let $f: \{0,1\}^E \to \mathbb{R}$ be a monotone $1$-Lipschitz function. For any $s \in (0,1]$, $t > 0$:
\[\Pr\left[f(\vb{X}^{(s)}) \geq \E[f(\vb{X})]+t\right] \leq e^{-st}.\]
\end{restatable}

A helpful comparison point is McDiarmid's inequality, which instead proves the following: $\Pr\left[f(\vb{X}) \geq \E[f(\vb{X})]+t\right] \leq e^{-2t^2/|E|}$. The distinctions are: (a) our concentration inequality is bicriterion -- we analyze $f(\vb{X}^{(s)})$ instead of $f(\vb{X})$, and (b) our concentration has an exponent of $-st$ instead of $-2t^2/|E|$. In particular, McDiarmid's inequality depends on the dimension $|E|$ and cannot possibly kick in for $t \ll \sqrt{|E|}$, whereas our concentration inequality can kick in for any $t > 1/s$. A representative example to have in mind might be $s = \sqrt{\log(1/\varepsilon)/\E[f(\vb{X})]}$ and $t = \sqrt{\log(1/\varepsilon)\E[f(\vb{X})]}$. This results in a tail probability of $\varepsilon$ for exceeding $\E[f(\vb{X})]$ by $\sqrt{\log(1/\varepsilon)}$ multiples of $\sqrt{\E[f(\vb{X})]}$, which is ``Chernoff-strength''. \emph{But}, this concentration holds only for $f(\vb{X}^{(s)})$, rather than $f(\vb{X})$. This suffices for our application.

To prove \autoref{thm:newconcentration}, we utilize the entropy method for self-bounding functions~\cite{BLM00,BLM03,MR06,BLM09} in an unconventional way. 
We give a more detailed technical overview in \autoref{subsec:concentration-technical-overview}.

\subsection{Related Work}
There are three strands of related work: prophet inequalities, concentration inequalities, and attempts to generalize $k$-uniform guarantees.

\paragraph*{Prophet Inequalities.} Prophet inequalities have a long history in Mathematics, Computer Science, and Operations Research. Representative results include Krengel and Sucheston's initial $\frac{1}{2}$-approximation~\cite{KrengelS78}, Samuel-Cahn's elegant thresholding strategy~\cite{Samuel-Cahn84}, Chawla et al.'s connection to Bayesian mechanism design~\cite{ChawlaHMS10}, Kleinberg and Weinberg's extension to matroids~\cite{KleinbergW12}, and Dutting et al.'s connection to Price of Anarchy~\cite{DuttingFKL20}. 

Of particular relevance to our work are prophet inequalities for $k$-uniform matroids. The first $1-O(\sqrt{\frac{\log k}{k}})$ approximation was developed by~\cite{HajiaghayiKS07}, and the first asymptotically tight $1-O(\frac{1}{\sqrt{k}})$ approximation was developed by~\cite{Alaei14}. Subsequent works achieve the same $1-O(\frac{1}{\sqrt{k}})$ approximation with sample access~\cite{AzarKW14}, the optimal OCRS~\cite{JWZ22}, or a simpler OCRS~\cite{DW24}. 
The most technically related paper to our work is~\cite{FSZ16}, whose OCRS framework we leverage.
It remains an open question whether the prophet inequality for $k$-fold matroid unions can be improved to $1-O(\frac{1}{\sqrt{k}})$.

\paragraph*{Concentration Inequalities.} The most related concentration inequalities fit the same framework but consider different $f$. McDiarmid's inequality~\cite{mcdiarmid1989method} holds for all $1$-Lipschitz $f$, Schechtman's inequality holds for $f$ that are subadditive~\cite{Schechtman03}, Bucheron et al. derive an inequality for $f$ that are self-bounding functions~\cite{BLM00}, and Vondr\'{a}k derives an inequality for $f$ that are fractionally subadditive~\cite{Von10}. These inequalities are commonly used across Theoretical Computer Science, and especially within combinatorial prophet inequalities and Bayesian mechanism design~\cite{RubinsteinW15,RubinsteinS17}.

\paragraph*{Generalizing $k$-uniform matroids.} Recent work of~\cite{CSZ24} considers (offline) contention resolution and correlation gap inequalities. Here too, guarantees for $k$-uniform matroids are significantly stronger than what is achievable for arbitrary matroids. Their work similarly extends guarantees achievable for $k$-uniform matroids to $k$-fold matroid unions. In comparison to our work: (a) the general motivation is the same -- both works seek to extend stronger guarantees for $k$-uniform matroids to more general settings, (b) the problems studied and technical aspects are orthogonal,\footnote{While in principle, contention resolution and online contention resolution may appear similar, the relevant techniques are fundamentally different with little overlap. Similarly, while correlation gap inequalities are sometimes a useful tool in prophet inequalities, in this case there is no overlap.} (c) our work also proposes a bicriterion concentration inequality.

Another generalization of $k$-uniform matroids are \emph{packing constraints}, where each element has a $d$-dimensional size in $[0,1]^d$ and one can accept a subset of elements if their size vectors sum to at most $k$ in every coordinate.
Packing constraints have been studied in various online settings, including secretary model \cite{KTRV14}, prophet model \cite{AFNJM12}, and mixed model \cite{AGMS22}.

\section{Preliminaries}
\label{sec:prelims}
\paragraph*{Prophet Inequalities.}

In the \emph{prophet inequality} problem, we are given a ground set of elements $E$, a downward-closed family of feasible sets $\cF \subseteq 2^E$, and a distribution $\cD_e$ associated with each element $e \in E$. 
Elements arrive in an adversarial order.\footnote{There are various adversarial models. The weakest is the \emph{fixed-order adversary}, which sets the arrival order offline, based solely on the distributions. The strongest is the \emph{almighty adversary}, which sets the arrival order online, with full knowledge of all realizations of randomness and the algorithm's past decisions. Our negative result in \autoref{sec:largegirth} applies to fixed-order adversary, while our positive result in \autoref{sec:ocrs} holds against almighty adversary.} 
As each element $e$ arrives, its value $v_e$, independently drawn from $\cD_e$, is revealed.
At this point, an irrevocable decision must be made whether to include $e$ in its output $A$, while keeping $A \in \cF$.

For $c \in [0,1]$, we say an online algorithm implies a $c$-\emph{competitive} prophet inequality for $\cF$, if for any distributions $\{\cD_e\}_{e \in E}$, 
\[\E\left[\sum_{e \in A} v_e\right] \ge c \cdot \E\left[\max_{S \in \cF} \sum_{e \in S} v_e\right]\]
where the expectation is taken with respect to random variables $\{v_e\}_{e \in E}$ and the internal randomness of the algorithm.

\paragraph*{Online Contention Resolution Schemes.} 
Given a ground set of elements $E$ and a downward-closed family of feasible sets $\cF \subseteq 2^E$, we define the polytope of $\cF$ as the convex hull of all characteristic vectors of feasible sets, i.e., $\cP_\cF = \conv(\{\mathbf{1}_F : F \in \cF\}) \subseteq [0,1]^E$.

An \emph{online contention resolution scheme} (OCRS) takes a vector $\vb{x} \in \cP_\cF$ as input. Let $R(\vb{x}) \subseteq E$ be a random set where each element $e \in E$ is in $R(\vb{x})$ independently with probability $x_e$. The OCRS sees membership in $R(\vb{x})$ of elements in $E$, arriving in an adversarial order; when each $e \in E$ arrives, if $e \in R(\vb{x})$ (i.e., $e$ is ``active''), the scheme must decide irrevocably whether to include $e$ in its output $A$, while keeping $A \in \cF$.

For $c \in [0,1]$, an OCRS is called $c$-\emph{selectable} for $\cF$, if for any $\vb{x} \in \cP_\cF$, 
\[\Pr[e \in A \mid e \in R(\vb{x})] \ge c ~~~~\forall e \in E\]
where $A\in \cF$ is the output of the OCRS, and the probability is measured with respect to $R(\vb{x})$ and internal randomness of the OCRS.
As shown in \cite{FSZ16}, a $c$-selectable OCRS directly implies a $c$-competitive prophet inequality.

\begin{lemma}[\cite{FSZ16}]
\label{lem:ocrs-implies-prophet}
    For a ground set $E$ and a family of feasible sets $\cF \subseteq 2^E$, a $c$-selectable OCRS for $\cF$ implies a $c$-competitive prophet inequality for $\cF$.
\end{lemma}

\paragraph*{Matroids.}
A \emph{matroid} $\cM=(E,\cI)$ is defined by a ground set of elements $E$ and a non-empty downward-closed family of independent sets $\cI \subseteq 2^E$ with the \emph{exchange property}, i.e., for every $A,B \in \cI$ where $|A| > |B|$, there exists an element $e \in A \setminus B$ such that $B \cup \{e\} \in \cI$. 
Given a matroid $\cM=(E,\cI)$, the following notations are used throughout the paper:
\begin{itemize}
    \item The \emph{rank} of a set $S \subseteq E$ is the size of the largest independent set contained in $S$: $\rank(S)=\max\{ |I| :  I \subseteq S, I \in \cI\}$.
    \item The \emph{span} of a set $S \subseteq E$ is the set of elements that is not independent from $S$: $\Span(S)=\{ e \in E : \rank(S) = \rank(S \cup \{e\})\}$.
    \item The \emph{restriction} of $\cM$ to a set $S \subseteq E$ is a matroid $\cM |_{S} = (S, \cI|_{S}) = (S, \{I \in \cI : I \subseteq S\})$.
    \item The \emph{girth} of $\cM$ is the size of the smallest dependent set: $\girth(\cM)=\min\{ |S| : S \subseteq E, S \notin \cI\}$.
\end{itemize}

Following are some special matroids that we will use later.

\begin{example}[Uniform matroid]
    A \emph{$k$-uniform matroid} $\cM=(E,\cI)$ is a matroid in which the independent sets are exactly the sets that contains at most $k$ elements for an integer $k\ge 1$, i.e, $\cI=\{I \subseteq E : |I|\le k\}$.
\end{example}

\begin{example}[Graphical matroid]
    A \emph{graphical matroid} $\cM=(E,\cI)$ is a matroid in which the independent sets are the forests in a given undirected graph $G=(V,E)$, i.e., $\cI = \{I \subseteq E : I~\text{is acyclic in}~G\}$.
\end{example}

We formally define $k$-fold matroid union as follows.

\begin{definition}[$k$-fold matroid union]
Given a matroid $\cM=(E,\cI)$ and an integer $k\ge 1$, the \emph{$k$-fold union} of $\cM$ is defined as $\cM^k=\underbrace{\cM \lor \cM \lor \cdots \lor \cM}_{\text{$k$ times}} = (E,\cI^k)$ where \[\cI^k=\{ I_1 \cup I_2 \cup\cdots \cup I_k : I_1, I_2,\dots,I_k \in \cI\}.\]
In other words, a set $I$ is independent in $\cM^k$ if and only if $I$ can be partitioned into at most $k$ independent sets in $\cM$. Note that $\cM^k$ remains a matroid by the closure property of matroid union.
\end{definition}

\section{Large Girth is Not Sufficient}
\label{sec:largegirth}
In this section, we prove that a large girth is not sufficient for matroids (specifically, graphical matroids) to have a prophet inequality with a competitive ratio better than $\frac{1}{2}$.

\restategirth*

To construct a hard instance, we start with a dense graph of large girth. We then transform the graph by splitting each edge $(v_1,v_2)$ into two edges $(v_1,u)$ and $(v_2,u)$, where $u$ is a newly introduced vertex.
We obtain the final hard instance of the prophet inequality problem
by embedding the hard instance of the single-item case into each of these edge pairs $(v_1,u), (v_2,u)$.

The hardness of this instance arises from the following observation: without accepting both edges in a pair, the instance essentially reduces to $|E|$ independent hard instances of the single-item case.
On the other hand, one can accept at most $|V|-1$ extra pairs of edges (in addition to $|E|$ single-item problems) at the same time without forming a cycle, which could not contribute a lot to the final solution because the graph is dense.

\begin{proof}[Proof of \autoref{thm:girth}]

We employ a construction of \cite{LUW95} which provides dense graphs of large girth. In particular, we will use that for any fixed $k$ there exists some arbitrarily large $n$ such that there is a graph $G_n$ on $n$ vertices with at least $n \log n$ edges and girth at least $k$.\footnote{In fact, \cite{LUW95} prove a significantly stronger result, but the weaker version stated above suffices for our purposes.}
Specifically, consider the graph $G_n$ with vertices $V(G_n) = \{v_1, v_2, \ldots, v_n\}$ and edges $E(G_n) = \{e_1, e_2, \ldots, e_m\}$, where $m \ge n \log n$ and each edge $e_i=(a(i),b(i))\in E(G_n)$ connects vertices $a(i)$ and $b(i)$ in $V(G_n)$. We construct a new graph $H_n$ with $n+m$ vertices as follows:
\begin{itemize}
    \item Begin with a set of $n+m$ vertices, labeled $V(H_n) := \{w_1, w_2, \ldots, w_n\} \sqcup \{u_1, u_2, \ldots, u_m\}.$
    \item For each edge $e_i$ in $G_n$ connecting $v_{a(i)}$ and $v_{b(i)}$, add in $H_n$ an edge between $u_i$ and $w_{a(i)}$ (call it $f_i$) as well as an edge between $u_i$ and $w_{b(i)}$ (call it $f_i'$).  
\end{itemize}
Hence $H_n$ has a total of $2m$ edges. For $1 \le i \le m$, let the associated random variable $X_{f_i}$ of $f_i$ be a constant $1$, and let the associated random variable $X_{f_i'}$ of $f_i'$ follow a distribution which takes a value of $\frac{1}{\varepsilon}$ with probability $\varepsilon$, and a value of $0$ with probability $1 - \varepsilon$. 
We consider an instance of the prophet inequality problem where the online algorithm is presented edges in the order $(f_1, f_1', f_2, f_2', \ldots, f_m, f_m')$.

We first lower bound $\OPT(H_n)$, the expected value the optimal offline algorithm gets on this instance. Note that an offline algorithm could simply look at each pair $\{f_i, f_i'\}$ and take whichever edge has higher realized weight; this cannot create a cycle because every edge selected will be incident to a vertex of degree 1. We hence have the bound 
\[\OPT(H_n) \ge \sum_{i=1}^m \left( \varepsilon \cdot \frac{1}{\varepsilon} + (1 - \varepsilon) \cdot 1 \right) = m (2 - \varepsilon).\]

Fix an online algorithm $\mathcal{A}$, and we now give an upper bound on its expected performance $\mathcal{A}(H_n)$ on the instance. The lower bound relies on the following observation.

\begin{claim}\label{claim:nplusoneedgescycle}
There are at most $n-1$ values of $i$ in $\{1, 2, \ldots, m\}$ such that $\mathcal{A}$ accepts both $f_i$ and $f_i'$.
\end{claim}
\begin{proof}
Suppose there are at least $n$ such values of $i$; call them $i_1$, $i_2$, $\ldots$, $i_n$. As the original graph $G$ has $n$ vertices, and a forest on $n$ vertices has at most $n-1$ edges, we clearly see that there is a cycle among $\{e_{i_1}, e_{i_2}, \ldots, e_{i_n}\}$. That however would imply there is a cycle in $H$; namely, follow the cycle that existed in $G$, but replace each edge $e_i$ with the edge $f_i$ followed by the edge $f_i'$.
\end{proof}

For each $1 \le i \le m$, we now consider cases for what $\mathcal{A}$ gets in expectation from $\{f_i, f_i'\}$ right after $f_i$ arrives:
\begin{itemize}
    \item If $\mathcal{A}$ rejects $f_i$, then it clearly gets in expectation at most $1$ from $\{f_i, f_i'\}$ because $\mathbb{E}[X_{f_i'}] = 1$.
    \item If $\mathcal{A}$ accepts $f_i$ and rejects $f_i'$, then it clearly gets weight at most $1$ from $\{f_i, f_i'\}$.
    \item If $\mathcal{A}$ accepts $f_i$ and accepts $f_i'$, then it clearly gets weight at most $1 + \frac{1}{\varepsilon}$ from $\{f_i, f_i'\}$.
\end{itemize}

Let $C_1$ denote the set of all $i \in [m]$ such that $\mathcal{A}$ rejects $f_i$, let $C_2$ denote the set of all $i \in [m]$ such that $\mathcal{A}$ accepts $f_i$ and rejects $f_i'$, and let $C_3$ denote the set of all $i \in [m]$ such that $\mathcal{A}$ accepts $f_i$ and $f_i'$. Note $C_1$, $C_2$, and $C_3$ are random (disjoint) sets that may depend on the values realized by $\{X_{f_i}, X_{f'_i}\}_{i=1}^m$ and any randomness in $\mathcal{A}$. By the above cases, we can see that in expectation, $\mathcal{A}$ gets score at most 
\[\sum_{i \in C_1} 1 + \sum_{i \in C_2} 1 + \sum_{i \in C_3} \left( 1 + \frac{1}{\varepsilon} \right) = |C_1| + |C_2| + |C_3| \cdot \left( 1 + \frac{1}{\varepsilon} \right).\]
Although $|C_1|$, $|C_2|$, and $|C_3|$ are random variables, $|C_1| + |C_2| \le m$ always, and by \autoref{claim:nplusoneedgescycle} we have $|C_3| \le n-1$ always. Hence, in expectation (averaging over all possible realizations of $C_1$, $C_2$, and $C_3$), we can bound the performance of $\mathcal{A}$ on $H_n$ by $\mathcal{A}(H_n) \le m + n \left( 1 + \frac{1}{\varepsilon} \right).$ As $n$ grows, we can compute 
\[\liminf_{n \rightarrow \infty} \frac{\mathcal{A}(H_n)}{\OPT(H_n)} \le \lim_{n \rightarrow \infty} \frac{m + n \left( 1 + \frac{1}{\varepsilon} \right)}{m(2-\varepsilon)} = \frac{1}{2 - \varepsilon}.\]
Taking $\varepsilon \rightarrow 0$ demonstrates the claimed result. 
\end{proof}

\section{$k$-Fold Unions are Sufficient}
\label{sec:ocrs}
Our main goal in the section is to construct a good OCRS for $k$-fold matroid unions (\autoref{thm:ocrs-for-k-fold-union}). Combining with the reduction from prophet inequalities to OCRSs by \cite{FSZ16} (\autoref{lem:ocrs-implies-prophet}), this immediately implies the existence of good prophet inequality for all $k$-fold matroid unions (\autoref{thm:kfoldunion}).

\begin{theorem}
\label{thm:ocrs-for-k-fold-union}
There exists a $(1-O(\sqrt{\frac{\log k}{k}}))$-selectable OCRS for any $k$-fold matroid union $\cM^k$.
\end{theorem}

Our OCRS for $k$-fold matroid unions builds on the chain decomposition approach used in the matroid OCRS by \cite{FSZ16}, outlined in \autoref{subsec:matroid-ocrs}.
We overview our approach and highlight main difficulties in \autoref{subsec:overview-our-ocrs}. The construction is then formally given and analyzed in \autoref{subsec:our-ocrs}, where the bicriterion concentration inequality in \autoref{sec:concentration} is used to bound its selectability.

\subsection{Recap: OCRS for general matroids}
\label{subsec:matroid-ocrs}

We briefly describe the idea of the $\frac{1}{4}$-selectable matroid OCRS by~\cite{FSZ16}. Specifically, they show that for any parameter $b \in (0,1)$, there exists a $(1-b)$-selectable OCRS for any matroid $\cM=(E,\cI)$ and $\vb{x}\in b\cdot \cP_\cM$.
Note that one can ``scale down'' a vector $\vb{x}$ from $\cP_\cM$ to $b \cdot \cP_\cM$ by only considering each element independently with probability $b$.
Formally:

\begin{fact}
\label{fact:shirnking}
For $b,c \in (0,1)$ and any matroid $\cM$, a $c$-selectable OCRS for all $\vb{x}\in b\cdot \cP_\cM$ implies a $bc$-selectable OCRS for all $\vb{x}\in \cP_{\cM}$.
\end{fact}

Therefore, it follows that a $b(1-b)$-selectable ORCS exists for any matroid $\cM$ and $\vb{x} \in \cP_\cM$.
By letting $b=\frac{1}{2}$, they obtain a $\frac{1}{4}$-selectable matroid OCRS.

\paragraph*{The greedy algorithm.}

Let us start with the simple greedy algorithm that always accepts the active element whenever possible.
When $\cM$ is a $1$-uniform matroid, the greedy algorithm is actually $(1-b)$-selectable for $\vb{x} \in b\cdot \cP_\cM$ (i.e., $\sum_{e \in E} x_e \le b$ since $\cM$ is $1$-uniform), since the selectability of an element $e \in E$ can be easily lower bounded as
\begin{align*}
  \Pr[\text{$e$ is accepted} \mid \text{$e$ is active}] 
  &\ge \Pr[\text{no other element is active} \mid \text{$e$ is active}] \\
  &\ge \Pr[\text{no element is active}].
\end{align*}
The first inequality holds because when there is no active elements besides $e$, the greedy algorithm can always accept $e$ even if it arrives at the end.
The second inequality holds due to the independence between elements.
Moreover, by Markov's inequality,
\begin{align*}
  \Pr[\text{no element is active}] 
  &= 1-\Pr[|R(\vb{x})| \ge 1]  \\
  &\ge 1- \E[|R(\vb{x})|] = 1-\sum_{e \in E}x_e \ge 1- b.  
\end{align*}
(Recall that $R(\vb{x})$ is the set of active elements.)

The first half of argument applies when $\cM$ is a general matroid: for every element $e \in E$,
\[\Pr[\text{$e$ is accepted} \mid \text{$e$ is active}] \ge \Pr[e\notin \Span(R(\vb{x}))].\]
However, unlike in $1$-uniform matroids, the probability that an element $e \in E$ is spanned by active elements $R(\vb{x})$ could be much smaller than $1-b$, even for a scaled $\vb{x}\in b\cdot\cP_\cM$. In fact, the selectability of the greedy algorithm can be arbitrarily bad for a general matroid $\cM$ (see, e.g., \cite{LS18}).

\paragraph*{Protection.}
Consider \autoref{alg:modified-greedy}, a modified greedy algorithm with a \emph{protection set} $S \subsetneq E$ that only handles elements in $E \setminus S$. Intuitively, the algorithm accepts every active element $e \in E \setminus S$ whenever it does not conflict with any element in $S$. As a result, elements in $S$ are ``prioritized'' over those in $E \setminus S$: regardless of which independent set from $S$ is accepted, it remains an independent set when combined with the accepted elements in $E \setminus S$.

\begin{algorithm}[H]
\caption{Modified greedy algorithm for $\cM=(E,\cI)$ with a protection set $S \subseteq E$}
\label{alg:modified-greedy}
\begin{algorithmic}
\State $A \gets \emptyset$ \Comment{the set of accepted elements in $E \setminus S$}
\For{each arriving active element $e \in E \setminus S$}
    \If{$e \notin \Span(A \cup S)$}
        \State $A \gets A \cup \{e\}$ \Comment{accepts element $e$}
    \EndIf
\EndFor
\end{algorithmic}
\end{algorithm}

For the modified greedy algorithm, we can similarly lower bound the selectability for $e \in E \setminus S$:
\[\Pr[\text{$e$ is accepted} \mid \text{$e$ is active}] \ge \Pr[e\notin \Span(R(\vb{x})\cup S)].\]
The good news is that, such probabilities can be further lower bounded by $1-b$ for the $S$ obtained using \autoref{alg:find-protection}, an iterative algorithm that updates $S$ by adding an element $e$ whenever $\Pr[e \in \Span(R(\vb{x}) \cup S)] > b$.%

\begin{algorithm}[H]
\caption{Find a protection set $S$ for $\cM=(E,\cI)$ and $\vb{x} \in b\cdot \cP_\cM$}
\label{alg:find-protection}
\begin{algorithmic}
\Function{Protect}{$\cM, \vb{x}, b$}
    \State $S \gets \emptyset$
        \While{$\exists e \in E \setminus S,~\Pr[e \in \Span(R(\vb{x}) \cup S)] > b$}
        \State $S \gets S \cup \{e\}$
    \EndWhile
    \State \Return $S$
\EndFunction
\end{algorithmic}
\end{algorithm}

Note that \autoref{alg:find-protection} always terminates since $E$ is a finite set, and the modified greedy algorithm with this protection set $S$ guarantees $(1-b)$-selectability for every element $e \in E \setminus S$. 
More importantly, the protection is non-trivial, i.e., $S$ is a proper subset of $E$. 

\begin{lemma}[\cite{FSZ16}]
\label{lem:protection-is-proper}
    For any matroid $\cM=(E,\cI)$ and $\vb{x}\in b\cdot\cP_\cM$, $\textproc{Protect}(\cM,\vb{x}, b)\subsetneq E$.
\end{lemma}

Therefore, it remains to get a good OCRS for $\cM|_S$ and $\vb{x}|_S$, the restriction of the original matroid and vector to the protection set $S$.

\paragraph*{Chain decomposition.}
The matroid OCRS in~\cite{FSZ16} starts with an offline prepossessing that finds the following \emph{chain decomposition} of the elements:
\[\emptyset = N_\ell \subsetneq N_{\ell-1} \subsetneq \cdots \subsetneq N_1 \subsetneq N_0 = E\]
where $N_{i+1} = \textproc{Protect}(\cM|_{N_i},\vb{x}|_{N_i}, b)$ for every $0 \le i < \ell$.
And the OCRS is then operates by invoking \autoref{alg:modified-greedy} on matroid $\cM |_{N_i}$ with a protection set $N_{i+1}$ for each $e \in N_i \setminus N_{i+1}$.

It is easy to see that these algorithms together produces an independent set of $\cM$, and the selectability for each element $e \in N_i \setminus N_{i+1}$ is
\begin{equation*}
  \Pr[\text{$e$ is accepted} \mid \text{$e$ is active}] \ge 1 - \Pr[e\in \Span_{\cM |_{N_i}}(R(\vb{x}|_{N_i}) \cup N_{i+1})] \ge 1-b
\end{equation*}
where the last inequality holds due to the way $N_{i+1}$ is obtained using \autoref{alg:find-protection}.
By setting $b=\frac{1}{2}$, the resulting OCRS is $\frac{1}{2}$-selectable given any matroid $\cM$ and $\vb{x} \in \frac{1}{2}\cdot\cP_\cM$.

\subsection{Overview of our construction}
\label{subsec:overview-our-ocrs}

We now give a high-level overview of our construction and highlight main difficulties.
Let us first examine the case when $\cM$ is a $k$-uniform matroid and see why the simple greedy algorithm works better for larger $k$.

\paragraph*{Intuition from $k$-uniform matroids.}
When $\cM$ is a $k$-uniform matroid, it turns out that the simple greedy algorithm that always accepts the active element whenever possible yields an OCRS with a selectability of $1-O(\sqrt{\frac{\log k}{k}})$.
To see this, consider the following tighter analysis of selectability for $k$-uniform matroids: for every element $e \in E$,
\[\Pr[\text{$e$ is accepted} \mid \text{$e$ is active}] \ge \Pr[e\notin \Span(R(\vb{x}))] = \Pr[|R(\vb{x})| < k].\]
Intuitively, $|R(\vb{x})|$ represents the number of slots occupied by active elements, and we know $|R(\vb{x})| < k$ indicates $e\notin \Span(R(\vb{x}))$.
We want the bad event $|R(\vb{x})| \ge k$ to occur with a small probability.

Note that $|R(\vb{x})|$ is a sum of Bernoulli random variables and it concentrates very well: if we consider a slightly scaled-down $\vb{x}\in (1-O(\sqrt{\frac{\log k}{k}}))\cdot \cP_\cM$, Chernoff bound (\autoref{thm:chernoff}) tells us that $\Pr[|R(\vb{x})| \ge k] \le \frac{1}{k}$.
By \autoref{fact:shirnking}, one can further derive an OCRS for $k$-uniform matroids with a selectability of $(1-O(\sqrt{\frac{\log k}{k}}))(1-\frac{1}{k}) = 1-O(\sqrt{\frac{\log k}{k}})$. 

To summarize, the greedy algorithm performs well on $k$-uniform matroids because of the existence of a \textbf{fine-grained occupancy indicator $|R(\vb{x})|$ that concentrates well}.

\paragraph*{Main idea and challenges.}

For a $k$-fold matroid union $\cM^k = (E, \cI^k)$, the simple greedy algorithm could perform very poor due to inherent non-uniformity of $\cM^k$.
In the matroid OCRS by \cite{FSZ16}, this is resolved using the idea of chain decomposition.
For each level, an iterative procedure (\autoref{alg:find-protection}) is used to find a protection set $S$ that includes all elements that are easily spanned by $R(\vb{x}) \cup S$.
This is done by directly looking at the probability $\Pr[e \in \Span(R(\vb{x}) \cup S)]$. 

Our idea is to construct a different chain decomposition based on functions $\omega_e(\cdot): 2^E \to [0,k]$ that act as a ``generalized occupancy indicator'' for each element $e$, such that $\omega_e(\emptyset)=0$, $\omega_e(S)=k$ if $e$ is spanned by the set $S$, and we want $\omega_e(\cdot)$ to be as smooth as possible (i.e., $1$-Lipschitz).
For each level of the chain decomposition, we will add $e$ to the protection set $S$ whenever the expected occupancy $\E[\omega_e(R(\vb{x}) \cup S)]$ is large.

For $k$-uniform matroids, a simple occupancy indicator would be $\omega_e(S)=\min(k,|S|)$ (since we require its value to be between $0$ and $k$).
However, extending the definition of an occupancy function to a general $k$-fold matroid union introduces several challenges:
\begin{enumerate}
    \item (\textbf{Compatibility with chain decomposition}) The most crucial part of the chain decomposition in \cite{FSZ16} is to show the protection set $S$ is always a proper subset of $E$ (\autoref{lem:protection-is-proper}).
    Similarly, we will need to show that it is always possible to find a protection set $S \subsetneq E$ such that the expected occupancy $\E[\omega_e(R(\vb{x}) \cup S)]$ for every $e \in E\setminus S$ is smaller than $k$ by a large enough margin.
    \item (\textbf{Chernoff-strength concentration})
    Based on the fact that $\E[\omega_e(R(\vb{x}) \cup S)]$ is sufficiently smaller than $k$, we ultimately want to show that $\Pr[\omega_e(R(\vb{x}) \cup S) = k]$ is very small, which would imply a good selectability for $e$.
    This is simple for $k$-uniform matroids by using Chernoff bound.
    However, it turns out $\omega_e(\cdot)$ for general $k$-fold matroid unions does not admit a standard Chernoff-strength concentration inequality, and much more efforts are required to achieve a similar selectability guarantee.
\end{enumerate}

\subsection{An OCRS for $k$-fold matroid unions}
\label{subsec:our-ocrs}

In \autoref{subsubsec:occupancy-function}, we define our candidate occupancy functions and show some useful properties.
Then, in \autoref{subsubsec:chain-decomposition}, we show these functions are compatible 
with the chain decomposition approach and can be used to get an OCRS for $k$-fold matroid unions.
Finally, in \autoref{subsubsec:selectability}, we prove the selectability of this OCRS by showing these functions concentrates well enough using \autoref{thm:newconcentration}.
Some proofs in the section are deferred to \autoref{app:missing-proofs} for ease of reading.

\subsubsection{The occupancy function}
\label{subsubsec:occupancy-function}

To define the occupancy function, we will instead work with the following \emph{extended $k$-fold unions} which essentially introduces $k$ parallel copies for each element.
They are still matroids, and OCRS for them implies OCRS for $k$-fold matroid unions.
Therefore, it suffices for us to give an OCRS for the extended $k$-fold union.

\begin{definition}[Extended $k$-fold union]
\label{def:extended-k-fold-union}
    Given a matroid $\cM=(E,\cI)$ and an integer $k\ge 1$, let $\cM_*=(E_*, \cI_*)$ be the matroid that contains $k$ parallel copies $(e,1),\ldots,(e,k)$ of each element $e \in E$. Formally,
    \begin{align*}
        E_* &= E \times [k] = \{(e, i) : e \in E, i \in [k]\},\\
        \cI_* &= \{\{(e_1,i_1),\ldots,(e_t,i_t)\} : \{e_1,\ldots,e_t\} \in \cI, i_1,\ldots,i_t \in [k]\}.
    \end{align*}
    And we define the \emph{extended $k$-fold union} $\cM^k_*=(E_*, \cI^k_*)$ of $\cM$ to be the $k$-fold union of $\cM_*$.

\end{definition}

\begin{lemma}
\label{lem:extended-k-fold-union-suffices}
    The extended $k$-fold union $\cM^k_*$ of a matroid $\cM$ is a matroid.
    Furthermore, a $c$-selectable OCRS for $\cM^k_*$ implies a $c$-selectable OCRS for $\cM^k$.
\end{lemma}

We are now ready to define the following \emph{occupancy function} on $\cM^k_*=(E_*=E\times [k],\cI^k_*)$.
Intuitively, the function indicates the number of ``slots'' for elements $(e,\cdot) \in E_*$ that are occupied by the elements in $S$.
We then show the occupancy function has good properties: it is monotone and 1-Lipschitz.
More importantly, the value of $\occ_e(S)$ can be used to deduce whether $(e,\cdot) \in E_*$ is spanned by other elements in $S$. %

\begin{definition}[Occupancy function]
    Given an extended $k$-fold union $\cM^k_*=(E_*=E\times [k],\cI^k_*)$, for every $e \in E$, define its \emph{occupancy function} $\occ_{e}: 2^{E_*} \to [0,k]$ as the function where for all $S \subseteq E_*$,\footnote{When it is clear from context, we will use $\rank(\cdot)$/$\Span(\cdot)$ to denote the rank/span of a set of elements in $\cM^k_*$ for the ease of notation.}
    \begin{equation*}
        \occ_e(S)=k-\rank(S \cup (\{e\}\times[k])) + \rank(S).
    \end{equation*}
\end{definition}

\begin{lemma}
\label{lem:occupancy-function-basic-property}
    For any extended $k$-fold union $\cM^k_*=(E_*,\cI^k_*)$ and element $(e,i) \in E_*$, $\occ_e$ satisfies
    \begin{enumerate}
        \item \emph{(Monotone)} $\occ_e(S) \le \occ_e(T)$ for every $S \subseteq T \subseteq E_*$;
        \item \emph{($1$-Lipschitz)} $\occ_e(S \cup \{a\}) - \occ_e(S) \le 1$ for every $S \subseteq E_*$ and $a \in E_*$.
    \end{enumerate}
\end{lemma}

\begin{lemma}
\label{lem:occupancy-function-good-indicator}
    For any extended $k$-fold union $\cM^k_*=(E_*,\cI^k_*)$, element $(e,i) \in E_* $, and set $S \subseteq E_*$, $\occ_e(S) < k$ implies $(e,i) \notin \Span(S \setminus \{(e,i)\})$.
\end{lemma}

\begin{example}
\label{example:k-uniform}
    When $\cM$ is a $1$-uniform matroid of size $n$, its extended $k$-fold union $\cM^k_*$ is a $k$-uniform matroid of size $kn$.
    For every $e \in E$ and $S \subseteq E_*$, we have
    \begin{align*}
        \rank(S \cup (\{e\}\times[k])) &= \min(k, |S \cup (\{e\}\times[k])|) = k,\\
        \rank(S) &= \min(k, |S|).
    \end{align*}
    Therefore, $\occ_e(S) = \min(k,|S|)$, i.e., the number of occupied slots by $S$.

    Also, note that for any $\vb{x_*} \in (1-O(\sqrt{\frac{\log k}{k}})) \cdot \cP_{\cM^k_*}$, the value $\occ_e(R(\vb{x_*}))$ concentrates very well as a capped sum over Bernoulli random variables.
    Therefore, the bad event $\occ_e(R(\vb{x_*}))=k$ rarely happens and the simple greedy algorithm without protection works.
\end{example}

\subsubsection{Chain decomposition based on occupancy functions}
\label{subsubsec:chain-decomposition}

Similar to the matroid OCRS by \cite{FSZ16}, our OCRS for extended $k$-fold union $\cM^k_*=(E_*,\cI^k_*)$ and $\vb{x_*} \in b\cdot \cP_{\cM^k_*}$ starts with an offline prepossessing step that finds the following chain decomposition of elements in $E_*$,
\[\emptyset = N_{\ell} \subsetneq N_{\ell-1} \subsetneq \cdots \subsetneq N_{1} \subsetneq N_{0} = E_*\]
where $N_{j+1} = \textproc{KFoldProtect}(\cM^k_*|_{N_j},\vb{x_*}|_{N_j}, b)$ for every $0 \le j < \ell$, as described in \autoref{alg:find-k-fold-protection}.
Unlike \autoref{alg:find-protection}, it relies on the occupancy functions which are only defined for extended $k$-fold unions.

\begin{algorithm}[H]
\caption{Find a protection set $S$ for extended $k$-fold union $\cM^k_*=(E_*,\cI^k_*)$ and $\vb{x_*} \in b\cdot \cP_{\cM^k_*}$}
\label{alg:find-k-fold-protection}
\begin{algorithmic}
\Function{KFoldProtect}{$\cM^k_*, \vb{x_*}, b$}
    \State $S_0,S \gets \emptyset$
    \While{$\exists e \in E \setminus S_0,~\E[\occ_e(R(\vb{x_*}) \cup S)] > b k$}
        \State $S_0 \gets S_0 \cup \{e\}$
        \State $S \gets S \cup (\{e\} \times [k])$
    \EndWhile
    \State \Return $S$
\EndFunction
\end{algorithmic}
\end{algorithm}

Before introducing our OCRS, we need to make sure the chain decomposition above is well-defined, i.e., \autoref{alg:find-k-fold-protection} will always returns a proper subset $S$ of elements, and $\cM^k_*|_{S}$ remains an extended $k$-fold union.
This is formally stated in \autoref{lem:k-fold-protection-is-proper}, which resembles \autoref{lem:protection-is-proper} in~\cite{FSZ16}.

\begin{lemma}
\label{lem:k-fold-protection-is-proper}
    For any $b \in (0,1)$, any extended $k$-fold union $\cM^k_*=(E_*,\cI^k_*)$ and $\vb{x_*}\in b\cdot\cP_{\cM^k_*}$, $S \subsetneq~E_*$ for $S = \textproc{KFoldProtect}(\cM^k_*,\vb{x_*}, b)$.
    Moreover, $\cM^k_*|_{S}$ remains an extended $k$-fold union.
\end{lemma}

Having obtained such a chain decomposition for $\cM^k_*$ and $\vb{x_*} \in b\cdot\cP_{\cM^k_*}$, our OCRS is simply running the modified greedy algorithm, \autoref{alg:modified-greedy}, for each submatroid $\cM^k_*|_{N_j}$ with a protection set $N_{j+1}$ for all $0 \le j < \ell$ together.
Note that although the chain decomposition is constructed with $\vb{x_*}$, an extra scaling factor of $e^{-(1-b)}$ will be applied before invoking \autoref{alg:modified-greedy}.
This will be useful later when we apply the bicriterion concentration inequality.

\begin{algorithm}[H]
\caption{OCRS for extended $k$-fold union $\cM^k_*=(E_*,\cI^k_*)$ and $\vb{x_*} \in b \cdot \cP_{\cM^k_*}$}
\label{alg:our-ocrs}
\begin{algorithmic}
    \State Construct the chain decomposition 
    $\emptyset = N_{\ell} \subsetneq \cdots \subsetneq N_{1} \subsetneq N_{0} = E_*$ for $\cM^k_*$ and $\vb{x_*}$
    \For{each arriving active element $(e,i) \in E_*$}
        \State Sample $r \sim \Ber(e^{-(1-b)})$ %
        \If{$r=0$}
            \State Reject $(e,i)$ 
        \Else
            \LComment{The set of remaining active elements follows the same distribution as $R(e^{-(1-b)} \vb{x_*})$.}
            \State Find $0 \le j < \ell$ such that $(e,i) \in N_j
        \setminus N_{i+1}$
            \State Invoke \autoref{alg:modified-greedy} for $\cM^k_*|_{N_j}$ with protection set $N_{j+1}$ for $(e,i)$
        \EndIf
    \EndFor
\end{algorithmic}
\end{algorithm}

The feasibility of such a scheme follows exactly from \cite{FSZ16} as running \autoref{alg:modified-greedy} on any chain decomposition always produces an independent set.
We are left to show the OCRS guarantees a good selectability for any $\cM^k_*$ and $\vb{x_*} \in b\cdot \cP_{\cM^k_*}$ for some parameter $b$.
In fact, we will set $b=1-\sqrt{\frac{\log k}{k}}$ and show the selectability is at least $1-O(\sqrt{\frac{\log k}{k}})$, proving \autoref{thm:ocrs-for-k-fold-union}.

\subsubsection{Analyzing the selectability}
\label{subsubsec:selectability}

Without loss of generality, let us focus on the selectability of elements in the first layer $E_* \setminus N_1$, since a same proof would work for all submatroid $\cM^k_* |_{N_j}$ as they remains to be extended $k$-fold unions.

By \autoref{lem:occupancy-function-good-indicator}, for every element $(e,i) \in E_* \setminus N_1$, its selectability can be lower bounded as
\begin{align*}
  \Pr[\text{$(e,i)$ is accepted} \mid \text{$(e,i)$ is active}] 
  &\ge \Pr[(e,i) \notin \Span((R(e^{-(1-b)}\vb{x_*}) \setminus \{(e,i)\}) \cup N_1)] \\
  &\ge \Pr[\occ_e(R(e^{-(1-b)}\vb{x_*}) \cup N_1) < k].  
\end{align*}
On the other hand, by the way chain decomposition is obtained using \autoref{alg:find-k-fold-protection}, we know even without the extra scaling of $e^{-(1-b)}$, the expected value of $\occ_e(R(\vb{x_*}) \cup N_1)$ is not too close to $k$:
\[\E[\occ_e(R(\vb{x_*}) \cup N_1)] \le bk.\]

For the ease of notation, denote $X=R(\vb{x_*})$ and $X'=R(e^{-(1-b)}\vb{x_*})$.
Fixing an element $(e,i) \in E_*$, define the function $f:2^{E_*}\to [0,k]$ where for every $S \subseteq E_*$,
\[f(S) = \occ_e(S \cup N_1).\]
Then, to lower bound selectability for $(e,i)$, it is equivalent to upper bound $\Pr[f(X') = k]$ given that $\E[f(X)] \le bk$.
Specifically, to get a selectability of $1-O(\sqrt{\frac{\log k}{k}})$, we will set $b=1-\sqrt{\frac{\log k}{k}}$, and it suffices to show the following \emph{bicriterion concentration inequality}:
\begin{equation}\label{eq:bizarre-concentration}
\E[f(X)] \le k - \sqrt{k \log k} \implies
 \Pr\left[f(X') \ge \E[f(X)]+ \sqrt{k \log k} \right] \le O\left(\frac{1}{k}\right). \tag{$\ast$}
\end{equation}

By \autoref{lem:occupancy-function-basic-property}, we know $f$ is always monotone and $1$-Lipschitz.
Then, using \autoref{thm:newconcentration} (and recall that $X'=R(e^{-\sqrt{\log k / k}} \vb{x_*})$), we have 
\[\Pr\left[f(X') \ge \E[f(X)]+ \sqrt{k \log k} \right] \le \exp\left( -\sqrt{\frac{\log k}{k}} \cdot \sqrt{k \log k}\right) = \frac{1}{k}.\]
Therefore, for extended $k$-fold union $\cM^k_*$ and $\vb{x_*} \in b\cdot\cP_{\cM^k_*}$, running \autoref{alg:our-ocrs} yields
\begin{align*}
    \Pr[\text{$(e,i)$ is accepted} \mid \text{$(e,i)$ is active}] 
    &\ge 1- \Pr\left[f(X') \ge k\right] \\
    &\ge 1- \Pr\left[f(X') \ge \E[f(X)]+ \sqrt{k \log k} \right] \ge 1-\frac{1}{k}.
\end{align*}
Together with \autoref{fact:shirnking} and \autoref{lem:extended-k-fold-union-suffices}, we prove \autoref{thm:ocrs-for-k-fold-union} by showing the existence of an OCRS for all $k$-fold union $\cM^k$ and $\vb{x_*} \in \cP_{\cM^k}$ with a selectability of
\[\left(1-\frac{1}{k}\right) \cdot b \cdot e^{-(1-b)} = \left(1-\frac{1}{k}\right)\cdot \left(1-\sqrt{\frac{\log k}{k}}\right) \cdot e^{-\sqrt{\frac{\log k}{k}}} = 1 - O\left(\sqrt{\frac{\log k}{k}}\right).\]

\begin{remark}
It might seems bizarre and unnecessary to consider $f(X')$ instead of $f(X)$.
Indeed, since $f$ is monotone non-decreasing, the following claim that only contains $f(X)$ would imply \eqref{eq:bizarre-concentration}, and it looks more like a standard concentration inequality:
\begin{equation}\label{eq:dream-concentration}
\E[f(X)] \le k - \sqrt{k \log k} \implies
 \Pr\left[f(X) \ge \E[f(X)]+ \sqrt{k \log k} \right] \le O\left(\frac{1}{k}\right). \tag{$\ast\ast$}
\end{equation}

We know \eqref{eq:dream-concentration} is true when $f$ is a sum over Bernoulli random variables by Chernoff bound, and it is tempting to use more powerful concentration inequalities to prove \eqref{eq:dream-concentration} for general $1$-Lipschitz $f$.
Unfortunately, such a bound does not exist for general monotone and $1$-Lipschitz set functions (see \autoref{sec:concentration} for details), and it turns out to be impossible even for the specific $f$ we use here, as \autoref{example:kn-uniform} shown.
\end{remark}

\section{A Bicriterion Concentration Inequality}
\label{sec:concentration}
In this section, we assume the ground set $E=[n]$ and consider a function $f:\{0,1\}^n \to \mathbb{R}$ that satisfies the following properties:\footnote{Note that $f$ can be equivalently viewed as a function over subsets of a ground set of size $n$, as we did in \autoref{sec:ocrs}.}
\begin{enumerate}
    \item (\textbf{Monotone}) $f(\vb{x}) \le f(\vb{y})$ for all $\vb{x},\vb{y} \in \{0,1\}^n$ where $\vb{x} \le \vb{y}$ (element-wise).
    \item (\textbf{$1$-Lipschitz}) $|f(\vb{x}) - f(\vb{y})| \le \lVert \vb{x}-\vb{y} \rVert_1$ for all $\vb{x},\vb{y} \in \{0,1\}^n$.
\end{enumerate}
Also, let $\vb{X} = (X_1, X_2, \ldots, X_n)$ be a vector of $n$ independent Bernoulli random variables where $X_i \sim \Ber(p_i)$ for each $i \in [n]$ and $\vb{p} \in [0,1]^n$.
For simplicity, we denote this as $\vb{X} \sim \Ber(\vb{p})$.

We are interested in how well $f(\vb{X})$ concentrates on its upper tail.
By McDiarmid's inequality (\autoref{thm:mcdiarmid}), for every $t > 0$,
\[\Pr\left[f(\vb{X}) \ge \E[f(\vb{X})]+t\right] \le e^{-\frac{2t^2}{n}},\]
Unfortunately, the bound depends on the dimension $n$, whereas our application in \autoref{sec:ocrs} requires a \emph{dimension-free} bound that is independent from $n$.
In fact, it is known that dimension-free concentration inequality does not exist for $f$ in general (see, e.g., \cite{Von10}).

The good news is that, for our application, it suffices to consider another $\vb{X'}\sim\Ber(\vb{p'})$ with slightly smaller parameters $\vb{p'}<\vb{p}$ and show $f(\vb{X'})$ does not exceed $\E[f(\vb{X})]$ by much, with high probability. Formally, we define $\vb{X}^{(s)}$ with a scaling factor $s$ as follows:

\begin{definition}[Scaling]
    Given $n$ independent Bernoulli random variables $\vb{X} \sim \Ber(\vb{p})$, for any \emph{scaling factor} $s \ge 0$, define $\vb{X}^{(s)} \sim \Ber(e^{-s}\vb{p})$.
    In other words, $X^{(s)}_i \sim \Ber(e^{-s}p_i)$ for all $i \in [n]$.
\end{definition}

And we prove \autoref{thm:newconcentration}, a \emph{bicriterion concentration inequality}, where the bound depends on both the scaling factor $s$ and the deviation size $t$.

\restateconcentration*

For our application in \autoref{sec:ocrs}, we basically set $s=\sqrt{\frac{\log k}{k}}, t=\sqrt{k\log k}$ for some $k \approx \E[f(\vb{X})]$ and the inequality gives us $\Pr[f(\vb{X}^{(s)}) \ge k + \sqrt{k \log k}] \le \frac{1}{k}$.
Note that this bound is sharp up to a constant factor in the exponent: even in the case where $f(\vb{x})= \sum_{i=1}^n x_i$, the Chernoff bound of $f(\vb{X}^{(s)})$ only yields $\Pr[f(\vb{X}^{(s)}) \ge k + \sqrt{k \log k}] \le O(\frac{1}{k^c})$ for some constant $c$.

\subsection{Technical overview}
\label{subsec:concentration-technical-overview}

Before getting into the proof, let us first outline our approach and highlight the main difficulty.
Our proof utilizes the entropy method for self-bounding functions \cite{BLM00,BLM03,MR06,BLM09}.
Roughly speaking, to prove a exponential concentration inequality for some $Z=f(\vb{X})$, the plan is to establish a differential inequality for the moment-generating function $\E[e^{\lambda Z}]$ based on the following modified logarithmic Sobolev inequality.
If this differential inequality implies strong bounds for $\E[e^{\lambda Z}]$, a concentration inequality can be subsequently obtained.

\begin{lemma}[A modified logarithmic Sobolev inequality \cite{Mas00}]
\label{lem:log-sobolev}
    Given $n$ independent Bernoulli random variables $\vb{X}$ and a function $f: \{0,1\}^n \to \mathbb{R}$.
    Let $Z=f(\vb{X})$ and $Z_i=f_i(X_1,\ldots,X_{i-1},X_{i+1},\ldots,X_n)$ for an arbitrary function $f_i:\{0,1\}^{n-1} \to \mathbb{R}$.
    For any $\lambda \in \mathbb{R}$,
    \[\lambda\E\left[Z e^{\lambda Z}\right] -\E\left[e^{\lambda Z}\right] \log \E\left[e^{\lambda Z}\right] \le \sum_{i=1}^n \E\left[e^{\lambda Z}\phi(-\lambda(Z-Z_i))\right]\]
    where $\phi(x)=e^x-x-1$.
\end{lemma}

Whether \autoref{lem:log-sobolev} can be effectively converted into a useful differential inequality for $\E[e^{\lambda Z}]$ depends on the choice of $\{Z_i\}_{i\in [n]}$.
For a monotone function $f$, a typical choice is $Z_i = f(X_1,\ldots,X_{i-1},0,X_{i+1},\ldots,X_n)$, and previous works have demonstrated that such a conversion is possible if $f$ is $1$-Lipschitz and the following condition holds almost surely for some constants $a,b\ge 0$:~\footnote{In this case, $f$ is a so-called \emph{$(a,b)$-self-bounding} function \cite{MR06, BLM09}.}
\begin{equation}\label{eq:self-bounding}
  \sum_{i=1}^n Z-Z_i \le aZ+b.
  \tag{$\dag$}
\end{equation}

Now, given $Z^{(s)}=f(\vb{X}^{(s)})$ under a scaling factor $s > 0$, one might attempt to similarly derive a differential inequality of $\E[e^{\lambda Z^{(s)}}]$ based on \autoref{lem:log-sobolev} if the condition \eqref{eq:self-bounding} can be satisfied.
In fact, if we define $Z_i^{(s)}=f(X_1^{(s)},\ldots,X_{i-1}^{(s)},0,X_{i+1}^{(s)},\ldots,X_n^{(s)})$, the following holds:%
\[\E\left[\sum_{i=1}^n Z^{(s)} - Z^{(s)}_i\right] = -\odv{}{s} \E\left[Z^{(s)}\right].\]
Thus, if $-\odv{}{s} \E[Z^{(s)}] \le aZ^{(s)}+b$, then \eqref{eq:self-bounding} holds in expectation for $Z^{(s)}$; otherwise, $E[Z^{(s)}]$ is decreasing rapidly with respect to $s$ at that point.

As a result, either there exists some $s^* \in (0,s)$ such that \eqref{eq:self-bounding}  holds \emph{in expectation} for $Z^{(s^*)}$, or $\E[Z^{(s)}]$ becomes significantly smaller than $\E[Z^{(0)}]$.
Intuitively, the latter case should directly imply a bicriterion concentration result, leaving only the former case to be addressed.\footnote{If we do not aim for an exponential tail bound, these observations indeed suffice to get a Chebyshev-type bicriterion concentration inequality for $Z^{(s)}$, by using Efron-Stein inequality to bound its variance in the former case.} 
However, it turns out that such a use of \autoref{lem:log-sobolev} crucially depends on \eqref{eq:self-bounding} holding \emph{almost surely}, which is not applicable to such $Z^{(s^*)}$ in the former case.\footnote{Specifically, applying the entropy method for $\E[e^{\lambda Z}]$ requires $\sum_{i=1}^n \E[ e^{\lambda Z} (Z-Z_i)] \le \E[e^{\lambda Z} (aZ+b)]$ for every $\lambda$, which might be false even if \eqref{eq:self-bounding} holds with very high probability.}

Given this limitation, rather than working with moment-generating functions directly, we propose an alternative approach. 
Our key idea is to relate \autoref{lem:log-sobolev} with the following unconventional function, defined for every $\lambda \ge 0$:
\[F(\lambda)=\E\left[ e^{\lambda Z^{(\lambda)}} \right].\]
Note that this is not a moment-generating function, as $\lambda$ here also serves as the scaling factor of $Z$, causing the random variable $Z^{(\lambda)}$ to change with it.
Surprisingly, we can obtain the following upper bound for the derivative of $F(\lambda)$ that aligns well with \autoref{lem:log-sobolev}.

\begin{lemma}\label{lem:derivative-upper-bound}
    Given $n$ independent Bernoulli random variables $\vb{X}$ and a monotone $1$-Lipschitz function $f: \{0,1\}^n \to \mathbb{R}$.
    For any $\lambda \in (0,1]$,
    \[F'(\lambda) \le \E\left[Z^{(\lambda)} e^{\lambda Z^{(\lambda)}} \right] - \frac{1}{\lambda} \sum_{i=1}^n \E\left[ e^{\lambda Z^{(\lambda)}} \phi(-\lambda(Z^{(\lambda)} - Z_i^{(\lambda)}))\right]\]
    where $Z^{(\lambda)}=f(\vb{X}^{(\lambda)})$, $Z_i^{(\lambda)}=f(X_1^{(\lambda)},\ldots,X_{i-1}^{(\lambda)},0,X_{i+1}^{(\lambda)},\ldots,X_n^{(\lambda)})$, and $F(\lambda)=\E[ e^{\lambda Z^{(\lambda)}} ]$.
\end{lemma}

By combining \autoref{lem:derivative-upper-bound} with \autoref{lem:log-sobolev}, we can conclude that for all $\lambda \in (0,1]$,
\[\lambda F'(\lambda) \le F(\lambda) \log F(\lambda).\]
Solving this differential inequality provides an upper bound for $F(\lambda)$. \autoref{thm:newconcentration} then follows by applying Markov's inequality to the random variable $e^{s Z^{(s)}}$.

\subsection{Proof of \autoref{thm:newconcentration}}

Let $Z^{(\lambda)}=f(\vb{X}^{(\lambda)})$ and $Z_i^{(\lambda)}=f(X_1^{(\lambda)},\ldots,X_{i-1}^{(\lambda)},0,X_{i+1}^{(\lambda)},\ldots,X_n^{(\lambda)})$ throughout the proof.
Given \autoref{lem:log-sobolev} and \autoref{lem:derivative-upper-bound}, it is not hard to show the bicriterion concentration inequality.

\begin{proof}[Proof of \autoref{thm:newconcentration}]
    For any $\lambda > 0$, we apply \autoref{lem:log-sobolev} to $Z^{(\lambda)}$ and $\{Z_i^{(\lambda)}\}_{i \in [n]}$ and obtain
    \[\lambda \E\left[Z^{(\lambda)} e^{\lambda Z^{(\lambda)}}\right] - \E\left[e^{\lambda Z^{(\lambda)}}\right] \log \E\left[e^{\lambda Z^{(\lambda)}}\right] 
        \le \sum_{i=1}^n \E\left[ e^{\lambda Z^{(\lambda)}}\phi(-\lambda(Z^\lambda-Z_i^{(\lambda)}))\right]\]
    where $\phi(x)=e^x-x-1$. Rearranging the inequality, we have
    \[\lambda \left(\E\left[Z^{(\lambda)} e^{\lambda Z^{(\lambda)}}\right] -
    \frac{1}{\lambda} \sum_{i=1}^n \E\left[ e^{\lambda Z^{(\lambda)}}\phi(-\lambda(Z^\lambda-Z_i^{(\lambda)}))\right] \right)
        \le \E\left[e^{\lambda Z^{(\lambda)}}\right] \log \E\left[e^{\lambda Z^{(\lambda)}}\right].\]
    Together with \autoref{lem:derivative-upper-bound}, this gives us the following differential inequality for $F(\lambda)=\E[ e^{\lambda Z^{(\lambda)}} ]$:
    \[\lambda F'(\lambda) \le F(\lambda) \log F(\lambda), \quad \forall \lambda \in (0,1].\]
    And by letting $G(\lambda)=\log F(\lambda)$, we can rewrite the inequality as 
    \[\lambda G'(\lambda) \le G(\lambda), \quad \forall \lambda \in (0,1].\]
    
    Note that $G_0(\lambda) = \lambda \E[Z^{(0)}]$ is a solution to $\lambda G'(\lambda)= G(\lambda)$ for $\lambda \in (0,1]$.
    Define $g(\lambda) = \frac{G(\lambda)-G_0(\lambda)}{\lambda}$ and we have
    \[g'(\lambda) = \frac{G'(\lambda) - G_0'(\lambda)}{\lambda} - \frac{G(\lambda)-G_0(\lambda)}{\lambda^2} = \frac{(\lambda G'(\lambda) -G(\lambda)) - (\lambda G'_0(\lambda)-G_0(\lambda))}{\lambda^2}\le 0.\]
    Also note that $\lim_{\lambda \to 0^+} \frac{G(\lambda)}{\lambda} = G'(0) = \frac{F'(0)}{F(0)} = \E[Z^{(0)}]$ (where the last equality holds by \autoref{lem:derivative-of-exp-scaled-expectation}), therefore $\lim_{\lambda \to 0^+} g(\lambda) = 0$.
    Combining this with $g' \le 0$, we conclude that $g$ is non-positive on $(0,1]$. In other words, for $\lambda \in (0,1]$,
    \[G(\lambda) \le G_0(\lambda) = \lambda \E[Z^{(0)}].\]
    
    Finally, by Markov's inequality, we conclude that for any $\lambda \in (0,1]$ and $t>0$,
    \begin{equation*}
      \Pr\left[Z^{(\lambda)} \ge \E[Z^{(0)}]+t\right] = \Pr\left[e^{\lambda Z^{(\lambda)}} \ge e^{\lambda (\E[Z^{(0)}]+t)}\right] \le \frac{\E[e^{\lambda Z^{(\lambda)}}]}{e^{\lambda (\E[Z^{(0)}]+t)}} \le e^{-\lambda t}. \qedhere  
    \end{equation*}
\end{proof}

We are left to prove \autoref{lem:derivative-upper-bound}.
Let us first compute $F'(\lambda)$ by definition.

\begin{lemma}\label{lem:derivative-of-exp-scaled-expectation}
    For any $\lambda\ge 0$, $F'(\lambda) =  \E[Z^{(\lambda)} e^{\lambda Z^{(\lambda)}}]-\sum_{i=1}^n \E[ e^{\lambda Z^{(\lambda)}}-e^{\lambda Z_i^{(\lambda)}}]$.
\end{lemma}

\begin{proof}
    Define a function $h:\mathbb{R}\times [0,1]^{n} \to \mathbb{R}$ as
    \[h(t,\vb{q})=\E_{\vb{Y} \sim \Ber(\vb{q})}\left[e^{t f(\vb{Y})}\right].\]
    For each $i \in [n]$ and $b \in \{0,1\}$, denote $f_{i,b}(\vb{Y}) = f(Y_1,\ldots,Y_{i-1},b,Y_{i+1},\ldots,Y_n)$
    and we can compute the partial derivative of $h$ with respect to $q_i$ as
    \begin{align*}
        \pdv{}{q_i} h(t, \vb{q}) 
        &= \pdv{}{q_i}\left(q_i \E_{\vb{Y} \sim \Ber(\vb{q})} [e^{t f_{i,1}(\vb{Y})}] + (1-q_i) \E_{\vb{Y} \sim \Ber(\vb{q})} [e^{t f_{i,0}(\vb{Y})}]\right)\\
        &=\E_{\vb{Y} \sim \Ber(\vb{q})} [e^{t f_{i,1}(\vb{Y})} - e^{t f_{i,0}(\vb{Y})}].
    \end{align*}

    Recall that $\vb{X} \sim \Ber(\vb{p})$ and $F(\lambda)=\E[ e^{\lambda Z^{(\lambda)}} ]=h(\lambda,e^{-\lambda}\vb{p})$. Therefore,
    \begin{align*}
        F'(\lambda) 
        &= \odv{t}{\lambda} \cdot \pdv{}{t} h(\lambda,e^{-\lambda}\vb{p}) +\sum_{i=1}^n \odv{q_i}{\lambda} \cdot \pdv{}{q_i} h(\lambda,e^{-\lambda} \vb{p})\\
        &=1 \cdot \E\left[f(\vb{X}^{(\lambda)}) e^{\lambda f(\vb{X}^{(\lambda)})}\right] + \sum_{i=1}^n (-e^{-\lambda} p_i) \cdot \E_{\vb{Y} \sim \Ber(e^{-\lambda} \vb{p})}\left[e^{\lambda f_{i,1}(\vb{Y})} - e^{\lambda f_{i,0}(\vb{Y})}\right]\\
        &=\E\left[f(\vb{X}^{(\lambda)}) e^{\lambda f(\vb{X}^{(\lambda)})}\right] - \sum_{i=1}^n \E_{\vb{Y} \sim \Ber(e^{-\lambda} \vb{p})} \left[e^{\lambda f(\vb{Y})} - e^{\lambda f_{i,0}(\vb{Y})}\right]\\
        &=\E\left[Z^{(\lambda)} e^{\lambda Z^{(\lambda)}}\right] - \sum_{i=1}^n \E\left[e^{\lambda Z^{(\lambda)}} - e^{\lambda Z_i^{(\lambda)}}\right]. \qedhere
    \end{align*}
\end{proof}

Then we further derive a lower bound to the latter term, $\sum_{i=1}^n \E[ e^{\lambda Z^{(\lambda)}}-e^{\lambda Z_i^{(\lambda)}}]$.

\begin{lemma}\label{lem:exp-self-bounding}
    For any $\lambda \ge 0$,
    $\sum_{i=1}^n \E[ e^{\lambda Z^{(\lambda)}} - e^{\lambda Z^{(\lambda)}_i}] \ge \lambda e^{-\lambda} \sum_{i=1}^n \E[ e^{\lambda Z^{(\lambda)}} (Z^{(\lambda)} - Z_i^{(\lambda)})].$
\end{lemma}

\begin{proof}
    We prove the inequality for each term separately and without expectation.
    For any $i \in [n]$, note that
    \[e^{\lambda Z^{(\lambda)}} - e^{\lambda Z^{(\lambda)}_i}=e^{\lambda Z^{(\lambda)}_i}(e^{\lambda (Z^{(\lambda)}-Z^{(\lambda)}_i)}-1) \ge e^{\lambda Z^{(\lambda)}_i} \cdot \lambda (Z^{(\lambda)}-Z^{(\lambda)}_i)\]
    since $e^x-1 \ge x$. 
    Meanwhile, we know $Z_i^{(\lambda)} \ge Z^{(\lambda)}-1$ by $1$-Lipschitzness of $f$. 
    Therefore, 
    \begin{equation*}
       e^{\lambda Z^{(\lambda)}} - e^{\lambda Z^{(\lambda)}_i} \ge \lambda e^{-\lambda}  \cdot e^{\lambda Z^{(\lambda)}} (Z^{(\lambda)} - Z_i^{(\lambda)}).  \qedhere
    \end{equation*}
\end{proof}

The following two facts of the function $\phi(x)=e^x-x-1$ will also be used.

\begin{fact}\label{fact:phi-by-lambda}
    For any $\lambda \in (0,1]$, $\frac{\phi(-\lambda)}{\lambda} \le \lambda e^{-\lambda}$.
\end{fact}

\begin{fact}\label{fact:phi-lambda-x}
    For any $\lambda \in \mathbb{R}$ and $x \in [0,1]$, $\phi(-\lambda x) \le \phi(-\lambda) x$.
\end{fact}

Now we are ready to prove the lemma.

\begin{proof}[Proof of \autoref{lem:derivative-upper-bound}]
    We upper bound $F'(\lambda)$ step-by-step as follows:
    \begin{align*}
      F'(\lambda) 
      &= \E[Z^{(\lambda)} e^{\lambda Z^{(\lambda)}}]-  \sum_{i=1}^n \E\left[ e^{\lambda Z^{(\lambda)}} - e^{\lambda Z_i^{(\lambda)}}\right] \tag{\autoref{lem:derivative-of-exp-scaled-expectation}}\\
      &\le \E[Z^{(\lambda)} e^{\lambda Z^{(\lambda)}}]- \lambda e^{-\lambda} \sum_{i=1}^n \E\left[ e^{\lambda Z^{(\lambda)}} (Z^{(\lambda)} - Z_i^{(\lambda)})\right] \tag{\autoref{lem:exp-self-bounding}} \\
      &\le \E[Z^{(\lambda)} e^{\lambda Z^{(\lambda)}}]- \frac{\phi(-\lambda)}{\lambda} \sum_{i=1}^n \E\left[ e^{\lambda Z^{(\lambda)}} (Z^{(\lambda)} - Z_i^{(\lambda)})\right] \tag{\autoref{fact:phi-by-lambda}} \\
      &\le \E[Z^{(\lambda)} e^{\lambda Z^{(\lambda)}}]- \frac{1}{\lambda} \sum_{i=1}^n \E\left[ e^{\lambda Z^{(\lambda)}} \phi(-\lambda(Z^{(\lambda)} - Z_i^{(\lambda)}))\right] \tag{\autoref{fact:phi-lambda-x}}
    \end{align*}
    where in the last step we also use the fact that $Z^{(\lambda)}-Z_i^{(\lambda)} \in [0,1]$, as $f$ is monotone and $1$-Lipschitz.
\end{proof}

\bibliography{main}

\appendix

\section{Useful Concentration Inequalities}
\label{app:concentration}
\begin{theorem}[Multiplicative Chernoff bound]
\label{thm:chernoff}
   Given $n$ independent Bernoulli random variables $X_1,X_2,\ldots,X_n$, let $X=\sum_{i=1}^n X_i$ denote their sum.
    For any $\delta > 0$, we have
    \[\Pr[X \ge (1+\delta)\E[X]] \le \exp\left(-\frac{\delta^2\E[X]}{2+\delta}\right).\]
\end{theorem}

\begin{theorem}[McDiarmid's inequality]
\label{thm:mcdiarmid}
    Given $n$ independent random variables $X_1,X_2,\ldots,X_n \in \cX$ and a function $f: \cX^n \to \mathbb{R}$. If for every $i \in [n]$ and $x_1,x_2,\ldots,x_n,x_i' \in \cX$, the function $f$ satisfies
    \[|f(x_1,\ldots,x_{i-1},x_i,x_{i+1},\ldots,x_n) - f(x_1,\ldots,x_{i-1},x_i',x_{i+1},\ldots,x_n)| \le c_i,\]
    then for any $t > 0$, we have
    \[\Pr[f(X) \ge \E[f(X)]+t] \le \exp\left(-\frac{2t^2}{\sum_{i=1}^n c_i^2}\right).\]
\end{theorem}

\section{Missing Proofs and Examples}
\label{app:missing-proofs}
\begin{proof}[Proof of \autoref{lem:extended-k-fold-union-suffices}]
    It is straightforward to check $\cM_*$ in \autoref{def:extended-k-fold-union} is a matroid, and hence its $k$-fold union $\cM^k_*$ remains a matroid by the closure property of matroid union.
    Also, note that the restriction of $\cM^k_*$ to $E \times \{1\}$, $\cM^k_* \mid_{E\times\{1\}}$, is isomorphic to the $k$-fold union $\cM^k$ of $\cM$, as there exists a simple bijection $(e,1) \mapsto e$ between $E \times \{1\}, \cI_* |_{E\times\{1\}}$ and $E, \cI$.
    Therefore, an $\alpha$-selectable OCRS for $\cM^k_*$ can also be used as an $\alpha$-selectable OCRS for $\cM^k$.
\end{proof}

\begin{proof}[Proof of \autoref{lem:occupancy-function-basic-property}]
    Note that the rank function for any matroid is a submodular function. 
    Therefore, $\rank(S \cup (\{e\}\times[k])) - \rank(S) \ge \rank(T \cup (\{e\}\times[k])) - \rank(T)$ for every $S \subseteq T$ by a simple induction, and thus $\occ_e(\cdot)$ is monotone.
    
    Also, we know the rank function is monotone,  and the rank of a set can increase by at most $1$ after adding an element. Therefore, $\rank(S \cup \{a\}) \ge \rank(S)$ and $\rank(S \cup \{a\} \cup (\{e\}\times[k])) \le \rank(S \cup (\{e\}\times[k])) + 1$ for every $a \in E_*$ and $S \subseteq E_*$. As a result, $\occ_e(\cdot)$ is 1-Lipschitz.
\end{proof}

\begin{proof}[Proof of \autoref{lem:occupancy-function-good-indicator}]
    When $\occ_e(S) < k$, we have $\rank(S \cup (\{e\} \times [k])) > \rank(S )$ and there exists (at least) one element $(e,j) \in \{e\} \times [k]$ such that $(e,j) \notin \Span(S)$. By definition of extended $k$-fold union, it further implies $(e,i) \notin \Span(S \setminus \{(e,i)\})$.
\end{proof}

\begin{proof}[Proof of \autoref{lem:k-fold-protection-is-proper}]
    It is straightforward to see $S=S_0 \times [k]$ when the algorithm terminates, and thus $\cM^k_*|_{S_0 \times [k]}$ is the extended $k$-fold union of $\cM|_{S_0}$ by definition.
    It remains to prove $S \subsetneq E_*$.
    Since $S$ must be a subset of the universe $E_*$, it suffices to show $S\ne E_*$.
    Our plan is to show that $S$ is not full rank in $\cM^k_*$, even after combined with all active elements $R(\vb{x_*})$ and take the expectation, i.e.,
    \[\E[\rank_{\cM^k_*}(R(\vb{x_*}) \cup S)]<\rank_{\cM^k_*}(E_*).\]
    This would directly imply $S\ne E_*$ by the monotonicity of the rank function.

    Denote $r=\rank_{\cM}(S_0)$.
    Let $e_1,e_2,\ldots,e_r \in S_0$ be the elements from $\cM$ that increase the rank of $S_0$ in $\cM$ during the execution of \autoref{alg:find-k-fold-protection}, and denote $e_i$ (for $1\le i \le r$) as the specific element that increases $\rank_\cM(S_0)$ from $i-1$ to $i$. 
    By definition, $\Span_\cM(\{e_1,e_2,\ldots,e_r\})=S_0$.
    In fact, we also have 
    \[\Span_{\cM^k_*}(\{e_1,e_2\ldots,e_r\} \times [k])= S.\]
    This is because $\{e_1,e_2,\ldots,e_r\} \times [k] \subseteq S$ is an independent set of size $kr$ in $\cM^k_*$ by definition of the extended $k$-fold union, and we can further show it is a basis of $S$.
    Suppose it is not, then there must be another independent set $T \subseteq S$ of size larger than $kr$.
    Since one can partition $T$ into $k$ disjoint independent sets $T_1, T_2, \ldots, T_k$ in $\cM_*$ where $\sum_{j\in [k]} |T_j| > kr$, we know there exists some $T_{j}$ of size larger than $r$, which leads to a contradiction as $T_{j} \subseteq S_0 \times [k]$ and $\rank_{\cM_*}(S_0 \times [k])=r$.
    
    Also, by the way the algorithm picks elements to be added to $S_0$, for every $1\le i\le r$ we have
    \[\E[\occ_{e_i}(R(\vb{x_*}) \cup (\{e_1,e_2,\ldots, e_{i-1}\} \times [k]))] > bk.\]
    Equivalently, we have
    \[\E[\rank_{\cM^k_*}(R(\vb{x_*}) \cup (\{e_1,e_2,\ldots, e_{i}\} \times [k])) - \rank_{\cM^k_*}(R(\vb{x_*}) \cup (\{e_1,e_2,\ldots, e_{i-1}\} \times [k]))] < (1-b)k.\]

    Together with these observations, we can upper bound $\E[\rank_{\cM^k_*}(R(\vb{x_*}) \cup S)]$ by a telescoping sum as follows:
    \begin{align*}
        \E[\rank_{\cM^k_*}(R(\vb{x_*}) \cup S)]
        &= \E[\rank_{\cM^k_*}(R(\vb{x_*}) \cup \{e_1,e_2,\ldots,e_r\} \times [k])]\\
        &=\E[\rank_{\cM^k_*}(R(\vb{x_*}))]+\sum_{i=1}^{r} \E[\rank_{\cM^k_*}(R(\vb{x_*}) \cup (\{e_1,e_2,\ldots,e_i\} \times [k]))\\
        &~~~~~~~~~~~~~~~~~~~~~~~~~~~~~~~~~~~~-\rank_{\cM^k_*}(R(\vb{x_*}) \cup (\{e_1,e_2,\ldots,e_{i-1}\})\times [k])]\\
        &< \E[\rank_{\cM^k_*}(R(\vb{x_*}))]+ (1-b)k r.
    \end{align*}
    The former term $\E[\rank_{\cM^k_*}(R(\vb{x_*}))]$ can be trivially upper bounded by $\E[|R(\vb{x_*})|]$ and further by $b\rank_{\cM^k_*}(E_*)$ due to $\vb{x_*} \in b\cdot\cP_{\cM^k_*}$.
    For the latter term involving $k r$, we already know $kr=\rank_{\cM^k_*}(S) \le \rank_{\cM^k_*}(E_*)$.
    In conclusion, we have
    \begin{equation*}
      \E[\rank_{\cM^k_*}(R(\vb{x_*}) \cup S)] < b\rank_{\cM^k_*}(E_*) + (1-b)\rank_{\cM^k_*}(E_*) = \rank_{\cM^k_*}(E_*). \qedhere  
    \end{equation*}
\end{proof}

\begin{example}[A counterexample to \eqref{eq:dream-concentration}]
\label{example:kn-uniform}
    Fix parameters $n,k$ where $n \gg k$, and consider the case when $\cM$ is an $n$-uniform matroid of size $2n$.
    Its extended $k$-fold union $\cM^k_*$ is a $kn$-uniform matroid of size $2kn$.
    Similar to \autoref{example:k-uniform}, for every $e \in E$ and $S \subseteq E_*$ we can derive
    \[\occ_e(S) = \begin{cases}
        0, & |S| \le kn - k\\
        |S|-(kn-k), &  kn - k < |S| < kn\\
        k, & |S| \ge kn.
    \end{cases}\]
    Since no protection is needed for uniform matroids, let $f(\cdot) = \occ_e(\cdot)$ for some fixed $e \in E$.
    When $\vb{x_*} = (\frac{1}{2}-\frac{1}{2n}) \cdot \mathbf{1}_{E_*}$ (namely, every element in $E_*$ is active with probability $\frac{1}{2}-\frac{1}{2n}$), $|X|$ will follow a binomial distribution with $kn-k$ as both its mean and median. As a result,
    \begin{align*}
        \E[f(X)] &\le k \Pr[f(X) > 0] = k \Pr\left[|X| > kn - k\right] \le \frac{k}{2},\\
        \text{while}~~\Pr[f(X) \ge k] &= \Pr\left[|X|\ge kn\right] \ge \Omega(1), \tag{$n \gg k$}
    \end{align*}
    which is a counterexample to the claim \eqref{eq:dream-concentration}.
    
    Note that this is not an actual counterexample to \autoref{alg:our-ocrs} (even without the extra scaling) since $\vb{x_*} \notin (1-O(\sqrt{\frac{\log k}{k}})) \cdot \cP_{\cM^k_*}$.
    But it shows that the condition $\E[f(X)] \le k - O(\sqrt{k \log k})$ alone is not enough to derive a good enough upper bound for $\Pr[f(X)\ge k]$, and it is crucial to also rely on the scaling applied to $\vb{x_*}$.
\end{example}

\end{document}